\documentclass[lettersize,journal]{IEEEtran}  \usepackage{amsmath,amssymb,graphicx}
\usepackage{geometry}
\usepackage{amsthm}
\usepackage{subfig}

\newtheorem{theorem}{Theorem}  
          
\newtheorem{proposition}{Proposition} 
\newtheorem{corollary}{Corollary} 
\theoremstyle{definition}
 
\newtheorem{remark}{Remark}

\title{Breaking Symmetry-Induced Degeneracy in Multi-Agent Ergodic Coverage via Stochastic Spectral Control}

\author{Kooktae Lee,~\IEEEmembership{Member,~IEEE}, and Julian Martinez
\thanks{K. Lee and J. Martinez are with the Department of Mechanical Engineering, 
New Mexico Institute of Mining and Technology, Socorro, NM 87801, USA 
(e-mail: sungjun.seo@student.nmt.edu; kooktae.lee@nmt.edu).}

\thanks{This work was supported by the U.S. National Science Foundation 
(NSF) CAREER Award under Grant CMMI-DCSD-2145810.}
}

\begin{document}

\maketitle

\begin{abstract}
Multi-agent ergodic coverage via Spectral Multiscale Coverage (SMC) provides a principled framework for driving a team of agents so that their collective time-averaged trajectories match a prescribed spatial distribution. While classical SMC has demonstrated empirical success, it can suffer from \emph{gradient cancellation}, particularly when agents are initialized near symmetry points of the target distribution, leading to undesirable behaviors such as stalling or motion constrained along symmetry axes. In this work, we rigorously characterize the initial conditions and symmetry-induced invariant manifolds that give rise to such directional degeneracy in first-order agent dynamics. To address this, we introduce a stochastic perturbation combined with a contraction term and prove that the resulting dynamics ensure almost-sure escape from zero-gradient manifolds while maintaining mean-square boundedness of agent trajectories. Simulations on symmetric multi-modal reference distributions demonstrate that the proposed stochastic SMC effectively mitigates transient stalling and axis-constrained motion, while ensuring that all agent trajectories remain bounded within the domain.
\end{abstract}


\section{Introduction}
Multi-agent ergodic coverage is key for applications such as surveillance, exploration, and environmental monitoring. Spectral Multiscale Coverage (SMC) \cite{mathew2009spectral,mathew2011metrics} drives a multi-agent system (MAS) so that collective time‑averaged trajectories match a prescribed spatial distribution. By encoding coverage objectives in the spectral domain, SMC provides a systematic mechanism for multi-scale distribution matching. 
Extensions of SMC include receding-horizon ergodic exploration for real-time coverage and target localization \cite{8114522} and decentralized ergodic control for distribution-driven multi-agent sensing \cite{abraham2018decentralized}. Further developments encompass heterogeneous multi-agent search \cite{sartoretti2021spectral}, coverage with low-information sensors \cite{coffin2022multi}, distributed coverage in unknown environments \cite{mantovani2025online}, and high-dimensional ergodic exploration \cite{shetty2021ergodic}, reflecting the active and ongoing application of spectral coverage methods in robotics and control.

While these methods have demonstrated empirical performance, rigorous asymptotic convergence guarantees for arbitrary initial conditions are still lacking.
Classical SMC can exhibit \emph{gradient cancellation}, leading to stalling or motion along symmetry axes, particularly when agents are initialized near symmetry points \cite{lee2025smcs}. Although the original SMC \cite{mathew2011metrics} derived the optimal control law for both first- and second-order dynamics, we here focus on first-order integrators:
\begin{equation}\label{eq:agent_dynamics}
\dot{\mathsf{x}}_i = u_i, \quad \mathsf{x}_i = [x_i, y_i]^{\top} \in \mathbb{R}^2, \quad i = 1,\dots,N,
\end{equation}
where $N$ denotes the total number of agents. 

Since the classical SMC frameworks \cite{mathew2009spectral}, \cite{mathew2011metrics} were originally developed for two-dimensional domains, we focus on the 2D case in this paper for clarity of exposition and analysis.
Nevertheless, the proposed stochastic spectral control formulation is not inherently restricted to two dimensions and can be extended to general $d$-dimensional domains in a straightforward manner without loss of generality.
Further, we assume that the domain $\Omega \subset \mathbb{R}^2$ is bounded, and that the reference distribution $\rho(\mathsf{x}) \ge 0$ is smooth with $\int_\Omega \rho(\mathsf{x})\,d\mathsf{x} = 1$.

SMC represents $\rho(\mathsf{x})$ using a Fourier basis $\{f_k\}_{k\in\mathbb{N}}$ in $L^2(\Omega)$ with Neumann boundary conditions. For a rectangular domain $\Omega = [0,L_x]\times[0,L_y]$, each basis function is
\begin{align}
f_k(x,y) &\equiv f_{mn}(x,y)\nonumber\\
&=
\cos\!\left(\frac{m\pi x}{L_x}\right)
\cos\!\left(\frac{n\pi y}{L_y}\right),
\, m,n\in\mathbb{N}_0,
\label{eq:f_k}
\end{align}
with Fourier coefficients constructed from the reference distribution $\rho(\mathsf{x})$
\begin{equation}
\mu_k = \int_\Omega f_k(\mathsf{x})\rho(\mathsf{x})\,d\mathsf{x}
\end{equation}
and with empirical time-averaged Fourier coefficients formed by agent trajectories
\begin{equation}
c_k(t)
=
\frac{1}{Nt}
\sum_{j=1}^N
\int_{0}^{t}
f_k(\mathsf{x}_j(\tau))\,d\tau.
\label{eq:c_k}
\end{equation}

Then, classical SMC drives agents to match $\{\mu_k\}$ via
\begin{align}
u_i(t)
&=
- u_{\max}\frac{B_i(t)}{\|B_i(t)\|},\label{eq:u_i}\\
B_i(t)
&\triangleq
\sum_{k} \lambda_k \nabla f_k(\mathsf{x}_i)
\big(c_k(t) - \mu_k\big)
\in \mathbb{R}^2,  
\end{align}
where $\lambda_k > 0$ and $u_{\max}$ denotes the input bound.

The key observation is that the control vector $B_i(t)$ may vanish for certain initial configurations due to cancellation among spectral gradients. 
In particular, specific initial agent placements can induce transient gradient cancellation, resulting in vanishing, near-vanishing control inputs, or  oscillatory motion, even in the absence of equilibria in low-density regions.

\section{Gradient Cancellation in Classical SMC}

\subsection{Example: Quadrimodal Symmetric Gaussian Reference}
\label{subsec:quadrimodal_example}

As illustrated in Fig. \ref{fig:axis_diagonal_examples}, consider a symmetric quadrimodal reference density on the rectangular domain
$\Omega=[0,L_x]\times[0,L_y]$,
defined as
\begin{equation}
\rho(x,y)
= \frac{1}{4}\sum_{k=1}^{4}\mathcal{N}(\mu_k,\Sigma),
\end{equation}
with Gaussian centers
\begin{equation*}
\begin{aligned}
\mu_1 &= [500,500]^{\top}, & \mu_2 &= [500,1500]^{\top},\\
\mu_3 &= [1500,500]^{\top}, & \mu_4 &= [1500,1500]^{\top}
\end{aligned}
\end{equation*}
and isotropic covariance
$
\Sigma=\sigma_{\rho}^2 I_2.
$

This density is symmetric with respect to the midlines
$x=L_x/2$ and $y=L_y/2$, and invariant under reflection across both coordinate axes.

\begin{figure}[t!]
    \centering
    \subfloat[]{
    \includegraphics[width=0.5\linewidth]{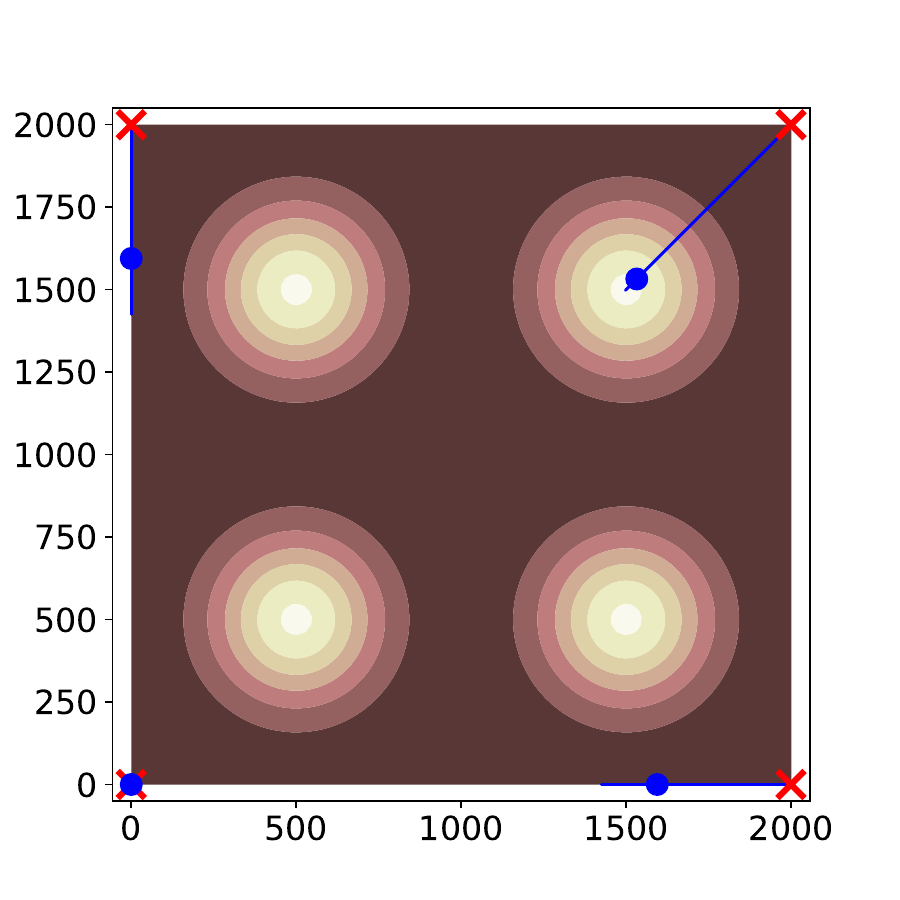}}
    \subfloat[]{
\includegraphics[width=0.5\linewidth]{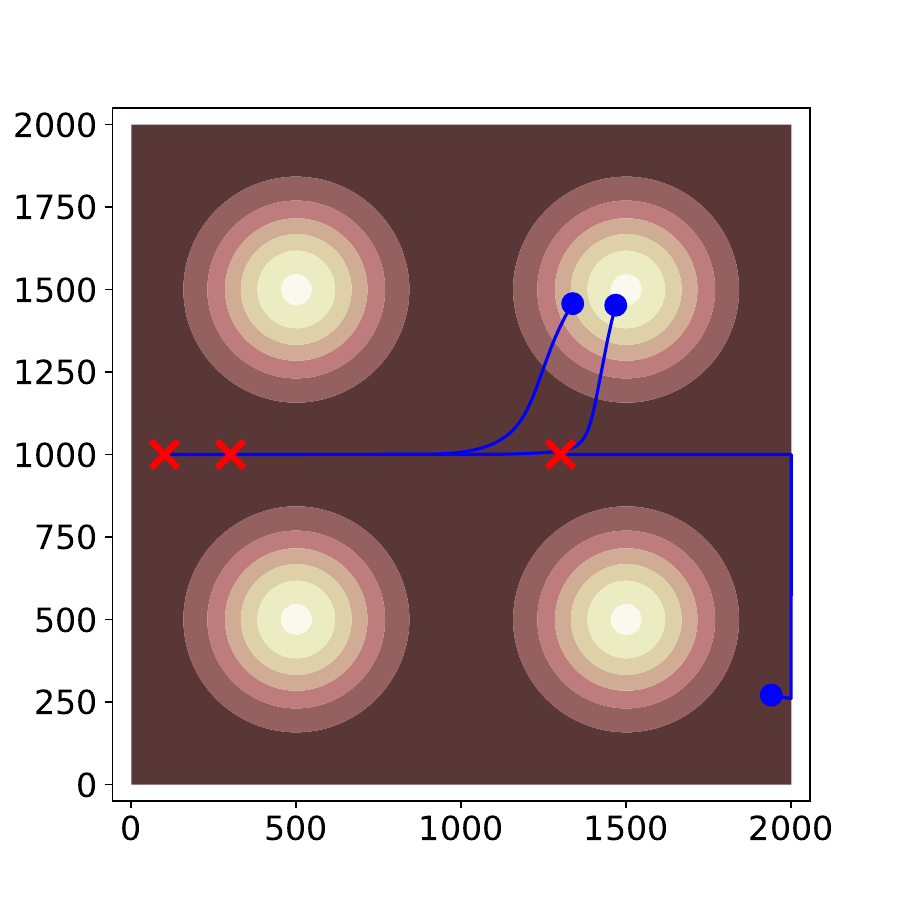}}
\caption{SMC directional degeneracy example. Crosses: initial positions; circles: final positions. 
(a) Four agents from corners: axis- or diagonal-constrained motion, with the origin agent stalled.  
(b) Three agents from non-symmetric positions: prolonged axis-aligned motion, with the left agents eventually increasing $y$ toward the upper-right Gaussian, and the right-most agent moving downward at the boundary before gradually decreasing $x$.}\label{fig:axis_diagonal_examples}
\end{figure}

\medskip
\begin{remark}[Axis-Constrained Motion and Stalling in SMC]
\label{rem:axis_constrained_motion}

Even though the spectral gradient $B_i(t)$ is nonzero at most points in $\Omega$, SMC dynamics can constrain agent motion along specific directions depending on initial positions and domain boundaries.

\begin{itemize}
    \item \textbf{Fig.~1(a) – Corner and diagonal initialization.}  
    Agents at domain corners or axes move along symmetry-constrained directions: $[0,L_y]^{\top}$ or $[L_x,0]^{\top}$ along axes, $[L_x,L_y]^{\top}$ along the diagonal. The origin $\mathsf{x} = [0,0]^{\top}$ remains stationary with $B_i(0)=0$. Even as the simulation time is extended arbitrarily long, agents initialized along symmetry axes or diagonals exhibit persistent stagnation due to gradient cancellation. 
    Depending on the initial configuration, the motion either stalls immediately or transitions, after a brief transient, into a near-stationary behavior characterized by only minute oscillations.

    \item \textbf{Fig.~1(b) – Non-symmetric initialization.}  
    Three agents at
    $\mathsf{x}_1 = [100,1000]^{\top}$, $\mathsf{x}_2 = [300,1000]^{\top}$, and $\mathsf{x}_3 = [1300,1000]^{\top}$    
    show constrained motion despite irregular placement. Left-most two agents initially move along $x$ and gradually gain $y$ motion toward a Gaussian peak. The right-most agent moves horizontally to $x=L_x=2000$, then vertically and finally horizontally in the negative $x$-direction. Early constrained motion is evident up to $T=150$\,s while full SMC convergence occurs much later.
\end{itemize}
\end{remark}

These examples show that axis- or diagonal-constrained, and stalled motions occur across a wide range of initial positions and boundaries. Importantly, this behavior can arise even when individual agent initial positions are not symmetric, indicating that it is an intrinsic property of the spectral gradient dynamics.

\subsection{Symmetry-Induced Invariant Manifolds in Spectral Coverage Control}

\begin{proposition}[Symmetry-Induced Invariant Sets in 2D SMC]
\label{prop:invariant_manifold_smc}
Consider the agent dynamics \eqref{eq:agent_dynamics} with control law \eqref{eq:u_i} on a rectangular domain $\Omega = [0,L_x]\times[0,L_y]$ using the Fourier basis \eqref{eq:f_k}, with a target density $\rho(\mathsf{x})$ that is symmetric about the vertical and horizontal midlines.  

Let an agent be initialized at a location $\mathsf{x}_i(0)$ that respects the same symmetry as $\rho$. Then, the following sets are invariant under the dynamics:
\[
\begin{aligned}
&\text{axes: } \mathcal{A}_0^x = \{\mathsf{x} \mid x=0\}, \quad \mathcal{A}_0^y = \{\mathsf{x} \mid y=0\}, \\
&\text{midlines: } \mathcal{A}_m^x = \{\mathsf{x} \mid x=L_x/2\}, \quad \mathcal{A}_m^y = \{\mathsf{x} \mid y=L_y/2\}, \\
&\text{diagonals: } \mathcal{D} = \{\mathsf{x} \mid y=x\}, \quad \mathcal{D}' = \{\mathsf{x} \mid y=L_x-x\}, \\
&\text{origin: } \mathsf{x} = [0,0]^\top.
\end{aligned}
\]

Specifically, for an agent starting on any of these sets with positions symmetric with respect to $\rho$, the velocity component normal to the set vanishes, ensuring that the agent remains confined to the set.
\end{proposition}

\begin{proof}
The gradient of the Fourier basis is
\[
\nabla f_{mn}(\mathsf{x}) = 
\begin{bmatrix} 
-\frac{m\pi}{L_x} \sin\frac{m\pi x}{L_x} \cos\frac{n\pi y}{L_y} \\[1ex] 
-\frac{n\pi}{L_y} \cos\frac{m\pi x}{L_x} \sin\frac{n\pi y}{L_y} 
\end{bmatrix}.
\]

\textbf{1. Axes and Origin:}  

- On $\mathcal{A}_0^x$ ($x=0$), $\partial_x f_{mn}([0,y]^\top) = 0$ for all $m,n$, hence $B_{i,x}=0$.  
- On $\mathcal{A}_0^y$ ($y=0$), $\partial_y f_{mn}([x,0]^\top) = 0$, hence $B_{i,y}=0$.  
- At the origin $\mathsf{x}=[0,0]^\top$, $\nabla f_{mn}(\mathbf{0})=0$ for all $m,n$, so $B_i=0$.

\textbf{2. Midlines and Diagonals:}  

- \emph{Midline $x=L_x/2$:}  
The Fourier basis satisfies
$\partial_x f_{mn}(\mathsf{x}) \propto \sin(m\pi x/L_x)$.  
For even $m$, $\sin(m\pi/2)=0$, hence the $x$-component of each mode vanishes identically along the midline.  
For odd $m$, the reference density is even with respect to the reflection $x \mapsto L_x-x$, which implies that the corresponding spectral coefficients vanish.  
Consequently, all contributions to the normal component cancel and $B_{i,x}=0$ holds for all $y$, yielding no velocity component transverse to the midline.

- \emph{Diagonal $y=x$:}  
Along the diagonal, the reference distribution is invariant under the coordinate exchange $(x,y)\mapsto(y,x)$, implying equality of the spectral-gradient components,
$B_{i,x}=B_{i,y}$.  
Therefore, the velocity component normal to the diagonal vanishes, and agent trajectories initialized on the diagonal remain confined to it.

Consequently, all sets stated above are invariant under the stated symmetry conditions.
\end{proof}

\begin{remark}[Interpretation of Symmetry-Induced Behavior]
Proposition~\ref{prop:invariant_manifold_smc} implies that if an agent is initialized exactly on a symmetry axis, midline, or diagonal, the normal component of the SMC vector field vanishes, and motion remains confined to that set. This does not imply attraction from generic initial conditions.  

Note that in practice, even extremely small normal components, negligible mathematically, can cause the agent to leave the invariant set because $B_i$ in \eqref{eq:u_i} is normalized. Once the agent moves slightly off the set, normalization amplifies the deviation, producing motion that appears to fully escape the invariant set, as shown in Fig.~\ref{fig:axis_diagonal_examples}(b), even though the invariant set holds rigorously in theory.
\end{remark}

\section{Stochastic Spectral Multiscale Control}

Classical SMC can exhibit \emph{directional degeneracy} when agents are initialized on symmetry manifolds of the reference distribution, where the spectral gradient vanishes along certain directions. This may cause transient or persistent behaviors such as stalling or motion constrained along symmetry axes. To address this, we introduce stochastic perturbations into the SMC dynamics. Modeling agent motion as an Itô process shows that additive noise breaks zero-gradient invariance while preserving spectral boundedness in expectation.

\subsection{Perturbed SMC Dynamics}
The SMC input \eqref{eq:u_i} is not Lipschitz continuous
at $B_i(t)=0$ due to the normalization operation.
While this does not pose difficulties in deterministic coverage
algorithms, it leads to analytical issues in the stochastic setting
when stochastic Lyapunov analysis and Itô calculus are employed to
establish boundedness properties.

To obtain a well-defined and Lipschitz continuous drift term, we
introduce the following smooth regularization:
\begin{equation}\label{eq:u_i_smooth}
u_i(t)
=
- u_{\max}
\frac{B_i(t)}{\sqrt{\|B_i(t)\|^2 + \varepsilon^2}},
\end{equation}
where $\varepsilon > 0$ is a small constant.
The control input \eqref{eq:u_i_smooth} remains uniformly bounded by
$u_{\max}$ and converges to the normalized input
\eqref{eq:u_i} as $\varepsilon \to 0$.

To prevent degeneracy, consider the perturbed SMC dynamics
\begin{equation}\label{eq:perturbed_smc_sde_corrected}
d\mathsf{x}_i(t) = u_i(t)\, dt + \sigma\, dW_i(t), \quad i=1,\dots,N,
\end{equation}
where $u_i(t)$ is the Lipschitz continuous SMC input in \eqref{eq:u_i_smooth} and $W_i(t)$ are independent $2$-dimensional standard Wiener processes.
Then, the we have the following result.

\medskip
\begin{theorem}[Almost-Sure Escape from Zero-Normal-Component Manifolds]
Consider $N$ agents evolving according to \eqref{eq:perturbed_smc_sde_corrected} on a bounded domain $\Omega \subset \mathbb{R}^2$ with $\sigma>0$. 
Define the aggregate state and input vectors as
\begin{align*}
X(t) &= [\mathsf{x}_1(t)^\top, \mathsf{x}_2(t)^\top, \dots, \mathsf{x}_N(t)^\top]^\top \in \mathbb{R}^{2N}, \\
U(t) &= [u_1(t)^\top, u_2(t)^\top, \dots, u_N(t)^\top]^\top \in \mathbb{R}^{2N}.
\end{align*}
Let 
$\Omega^N := \underbrace{\Omega \times \Omega \times \cdots \times \Omega}_{N \text{ times}} \subset \mathbb{R}^{2N}$,
and define the set of zero normal-component points for agent $i$ as
\[
\mathcal{M}_i := \{ X \in \Omega^N \mid B_i(X) \cdot \mathbf{n}_i(X) = 0 \},
\]
where $\mathbf{n}_i(X)$ denotes the unit normal vector to the relevant zero-gradient manifold for agent $i$ at state $X$, 
induced by the symmetry of the reference distribution (e.g., axes or diagonals in $\Omega$)

Then, for any initial condition $X(0) = X^0 \in \Omega^N$ and any $t>0$, it holds that
\[
\Pr(X(t) \in \mathcal{M}_i) = 0.
\]
\end{theorem}

\begin{proof}
Since $\Sigma = \sigma I_{2N}\succ 0$, the SDE
\[
dX(t) = U(t)\,dt + \Sigma\,dW(t)
\]
is a uniformly elliptic Itô diffusion on $\mathbb{R}^{2N}$ with smooth drift and diffusion coefficients.

By standard results in Malliavin calculus \cite{nualart2006malliavin}, for uniformly elliptic SDEs the solution $X(t)$ admits a smooth density $p(\cdot,t)$ with respect to the Lebesgue measure on $\mathbb{R}^{2N}$ for each $t>0$.  
This follows from H\"{o}rmander‑type smoothness conditions or uniform ellipticity results in Malliavin calculus (e.g., Malliavin’s absolute continuity lemma; see \cite{nualart2006malliavin} and references therein).

Each condition $B_i(X)\cdot \mathbf{n}_i(X) = 0$ defines a smooth hypersurface in $\mathbb{R}^{2N}$ of codimension 1. 
Hence, each hypersurface has Lebesgue measure zero. 
Since $\mathcal{M}_i$ is the union of all such hypersurfaces corresponding to agent $i$, it also has Lebesgue measure zero. 
Consequently, $\mathcal{M}_i$ is a Lebesgue null set in $\mathbb{R}^{2N}$.

Since the law of $X(t)$ has a smooth density with respect to Lebesgue measure, the probability of $X(t)$ lying in any Lebesgue null set is zero.  
Specifically,
\[
\Pr(X(t)\in\mathcal{M}_i) = \int_{\mathcal{M}_i}p(\mathsf{x},t)\,d\mathsf{x} = 0,
\]
because the integral of a smooth density over a set of measure zero is zero.  

Thus, for any $t>0$ it holds that $\Pr(X(t)\in\mathcal{M}_i)=0$, which establishes almost-sure escape from any zero-normal-component manifold.
\end{proof}

\begin{remark}[Practical Implementation]
When an agent starts on an invariant set (e.g., at the origin), adding stochastic perturbations may push it outside the bounded domain. In practice, this can be prevented using simple projection or reflection at the domain boundaries.
\end{remark}

\subsection{Mean-Square Boundedness of Perturbed Multi-Agent SMC Dynamics}
While stochastic perturbations are effective for breaking directional degeneracy, they can in principle drive agent trajectories away from nominal paths, raising stability concerns. To address this, we introduce a deterministic contraction term in the SMC dynamics, which ensures that trajectories remain bounded in expectation. This enables a rigorous analysis of \emph{infinite-time mean-square boundedness} (MSB), providing a theoretical guarantee not covered in prior SMC studies.

\medskip
\begin{theorem}[Infinite-Time Mean-Square Boundedness of Perturbed Multi-Agent SMC
with Contraction and Constant Noise]\label{thm:msb}
Consider the perturbed multi-agent SMC dynamics modeled as an It\^o stochastic
process with a contraction term and constant additive noise:
\begin{equation}\label{eq:smc_contraction_const}
d\mathsf{x}_i = u_i(t)\, dt - k \mathsf{x}_i dt + \sigma\, dW_i, \quad i = 1, \dots, N,
\end{equation}
where the control input $u_i(t)\in\mathbb{R}^2$ in \eqref{eq:u_i_smooth} satisfies
\begin{equation}\label{eq:u_bound}
\|u_i(t)\| < u_{\max}, \qquad \forall t \ge 0,
\end{equation}
$k>0$ is the contraction gain, $\sigma>0$ is the constant noise magnitude,
$W_i$ are independent standard Wiener processes in $\mathbb{R}^2$, and
$\mathsf{x}_i(0)\in \Omega$, a bounded domain.

Then, each trajectory $\mathsf{x}_i(t)$ satisfies
\[
\sup_{t \ge 0} \mathbb{E}[\|\mathsf{x}_i(t)\|^2] < \infty,
\]
i.e., the system is mean-square bounded (MSB) for all time.
\end{theorem}

\begin{proof}
To establish MSB, we define a standard quadratic Lyapunov function for each agent:
\[
V_i(\mathsf{x}_i) = \|\mathsf{x}_i\|^2 = \mathsf{x}_i^\top \mathsf{x}_i,
\]
which is positive definite and radially unbounded. 

Since the input \eqref{eq:u_i_smooth}
is Lipschitz continuous in \(\mathsf{x}_i\) due to the small regularization \(\varepsilon>0\), and the diffusion term is constant, 
the standard conditions for Itô's lemma is applied along with \eqref{eq:smc_contraction_const} as follows:
\begin{align*}
dV_i &= 2 \mathsf{x}_i^\top d\mathsf{x}_i + d\mathsf{x}_i^\top d\mathsf{x}_i \\
&= 2 \mathsf{x}_i^\top \Big(u_i - k \mathsf{x}_i \Big) dt + 2 \mathsf{x}_i^\top \sigma dW_i + d\mathsf{x}_i^\top d\mathsf{x}_i.
\end{align*}
The quadratic variation term is rigorously computed by including all drift and diffusion components:
\[
\begin{aligned}
d\mathsf{x}_i^\top d\mathsf{x}_i
&= (u_i dt - k \mathsf{x}_i dt + \sigma dW_i)^\top (u_i dt - k \mathsf{x}_i dt + \sigma dW_i) \\
&= (u_i dt)^\top (u_i dt) + (-k \mathsf{x}_i dt)^\top (-k \mathsf{x}_i dt)\\ 
&\quad + (\sigma dW_i)^\top (\sigma dW_i)
 + 2 (u_i dt)^\top (-k \mathsf{x}_i dt)\\
&\quad + 2 (u_i dt)^\top (\sigma dW_i)
 + 2 (-k \mathsf{x}_i dt)^\top (\sigma dW_i).
\end{aligned}
\]
Applying standard It\^o calculus rules:
\[
dt^2 = 0, \quad dt \, dW_i = 0, \quad dW_i^\top dW_i = \sum_{j=1}^{2} (dW_i^{(j)})^2 = 2\, dt,
\]
all terms involving products of $dt$ vanish, including the drift contributions. 
Therefore, only the diffusion term survives
\[
d\mathsf{x}_i^\top d\mathsf{x}_i = (\sigma dW_i)^\top (\sigma dW_i) = 2 \, \sigma^2 \, dt,
\]
where the factor $2$ arises from the two-dimensional state of agent $i$.

Substituting back gives
\begin{align*}
dV_i &= 2 \mathsf{x}_i^\top u_i dt -2 k \|\mathsf{x}_i\|^2 dt + 2 \sigma \mathsf{x}_i^\top dW_i + 2\, \sigma^2 dt.
\end{align*}
Taking expectation and noting that $\mathbb{E}[\mathsf{x}_i^\top dW_i] = 0$ for It\^o integrals, we obtain
\begin{align}
\frac{d}{dt} \mathbb{E}[V_i] = 2  \mathbb{E}\left[ \mathsf{x}_i^\top u_i \right] - 2 k \mathbb{E}[\|\mathsf{x}_i\|^2] + 2 \, \sigma^2.\label{eq:ddt_E[V_i]}
\end{align}
Next, we bound the cross term using Young's inequality. For any $\eta>0$,
\[
2 \mathsf{x}_i^\top u_i \le \eta \|\mathsf{x}_i\|^2 + \frac{1}{\eta}\|u_i\|^2.
\]
Choosing $\eta = k$, where $k>0$ is the contraction gain in \eqref{eq:smc_contraction_const} and using the uniform bound $\|u_i(t)\| < u_{\max}$ from
\eqref{eq:u_bound}, we obtain
\[
2 \mathbb{E}[\mathsf{x}_i^\top u_i]
< k \mathbb{E}[\|\mathsf{x}_i\|^2] + \frac{u_{\max}^2}{k}.
\]
Substituting this bound into \eqref{eq:ddt_E[V_i]} yields
\[
\frac{d}{dt} \mathbb{E}[V_i]
< -k \mathbb{E}[\|\mathsf{x}_i\|^2] + \frac{u_{\max}^2}{k} + 2 \, \sigma^2.
\]
Since $\mathbb{E}[\|\mathsf{x}_i\|^2] = \mathbb{E}[V_i]$, this implies the linear
differential inequality
\[
\frac{d}{dt} \mathbb{E}[V_i]
< -k \mathbb{E}[V_i] + \frac{u_{\max}^2}{k} + 2 \, \sigma^2.
\]
By comparison with the scalar ODE
\[
\dot Y = -k Y + \frac{u_{\max}^2}{k} + 2 \, \sigma^2,
\]
it follows that
\[
\sup_{t \ge 0} \mathbb{E}[V_i(t)]
<
\max\!\left\{
\mathbb{E}[V_i(0)],
\frac{u_{\max}^2}{k^2} + \frac{2 \, \sigma^2}{k}
\right\}
< \infty.
\]
Therefore,
\[
\sup_{t \ge 0} \mathbb{E}[\|\mathsf{x}_i(t)\|^2] < \infty,
\]
which establishes infinite-time mean-square boundedness.
\end{proof}

Theorem~\ref{thm:msb} establishes mean-square boundedness for the perturbed SMC dynamics, a result that is novel and not found in the original SMC formulation or its follow-up studies. By setting the stochastic term to zero, the analysis directly applies to the deterministic SMC dynamics, showing that the contraction term alone ensures the agent states remain bounded. The following corollary formalizes this boundedness property in the deterministic case.

\medskip
\begin{corollary}[Deterministic Contraction Guarantees Uniform Boundedness]
Consider the agent dynamics \eqref{eq:smc_contraction_const} with $\sigma = 0$.
Then, each agent state $\mathsf{x}_i(t)$ satisfies
\[
\|\mathsf{x}_i(t)\|
\le
\|\mathsf{x}_i(0)\| e^{-k t}
+ \frac{u_{\max}}{k} \big(1 - e^{-k t}\big),
\quad \forall t \ge 0,
\]
and consequently,
\[
\|\mathsf{x}_i(t)\| \le \max\Big\{ \|\mathsf{x}_i(0)\|, \frac{u_{\max}}{k} \Big\}, \quad \forall t \ge 0.
\]
Hence, the deterministic system is uniformly bounded for all time.
\end{corollary}

\begin{proof}
With $\sigma = 0$, the dynamics reduce to
\[
\dot{\mathsf{x}}_i = u_i(t) - k \mathsf{x}_i, \quad \|u_i(t)\| < u_{\max}.
\]
Taking norms and applying Cauchy--Schwarz gives
\[
\frac{d}{dt}\|\mathsf{x}_i\| = \frac{\mathsf{x}_i^\top \dot{\mathsf{x}}_i}{\|\mathsf{x}_i\|} \le \|u_i(t)\| - k \|\mathsf{x}_i\| < u_{\max} - k \|\mathsf{x}_i\|.
\]
Using the standard comparison principle, the scalar ODE
\[
\dot{\mathsf{y}} = u_{\max} - k \mathsf{y}, \quad \mathsf{y}(0) = \|\mathsf{x}_i(0)\|
\]
provides an upper bound for $\|\mathsf{x}_i(t)\|$. Solving this ODE yields
\[
\|\mathsf{x}_i(t)\| \le \|\mathsf{x}_i(0)\| e^{-k t} + \frac{u_{\max}}{k} \big(1 - e^{-k t}\big),
\quad \forall t \ge 0.
\]
Since the sum of a monotone decreasing term and a monotone increasing term achieves its maximum at either $t=0$ or $t \to \infty$, we further have
\[
\|\mathsf{x}_i(t)\| \le \max\Big\{ \|\mathsf{x}_i(0)\|, \frac{u_{\max}}{k} \Big\}, \quad \forall t \ge 0.
\]
Thus, the deterministic system is uniformly bounded for all time.
\end{proof}

\section{Simulation Results}

We revisit the quadrimodal reference scenario of Fig.~\ref{fig:axis_diagonal_examples} under the perturbed SMC framework. The simulations use identical initial agent positions and domain settings as the previous case, with parameters $\varepsilon=1\times 10^{-3}$,
$\sigma = 1\times 10^{-5}$, and $k = 1\times 10^{-3}$.

The results are shown in Fig.~\ref{fig:corrected}. Two key observations are evident. First, the undesired behaviors observed in classical SMC, such as stalling or motion constrained along symmetry axes, are completely eliminated by introducing stochastic perturbations, which allow agents initially trapped on invariant sets to escape, as guaranteed by Theorem~1. Second, despite the presence of these perturbations, the agents do not diverge outside the domain. This is consistent with Theorem~2, which ensures mean-square boundedness of the trajectories, guaranteeing that all agent motions remain bounded.

Overall, the simulations validate that the proposed stochastic SMC maintains spectral coverage objectives while effectively mitigating gradient cancellation and transient degeneracy.

\begin{figure}[!t]
\centering
\subfloat[]{
\includegraphics[width=0.5\linewidth]{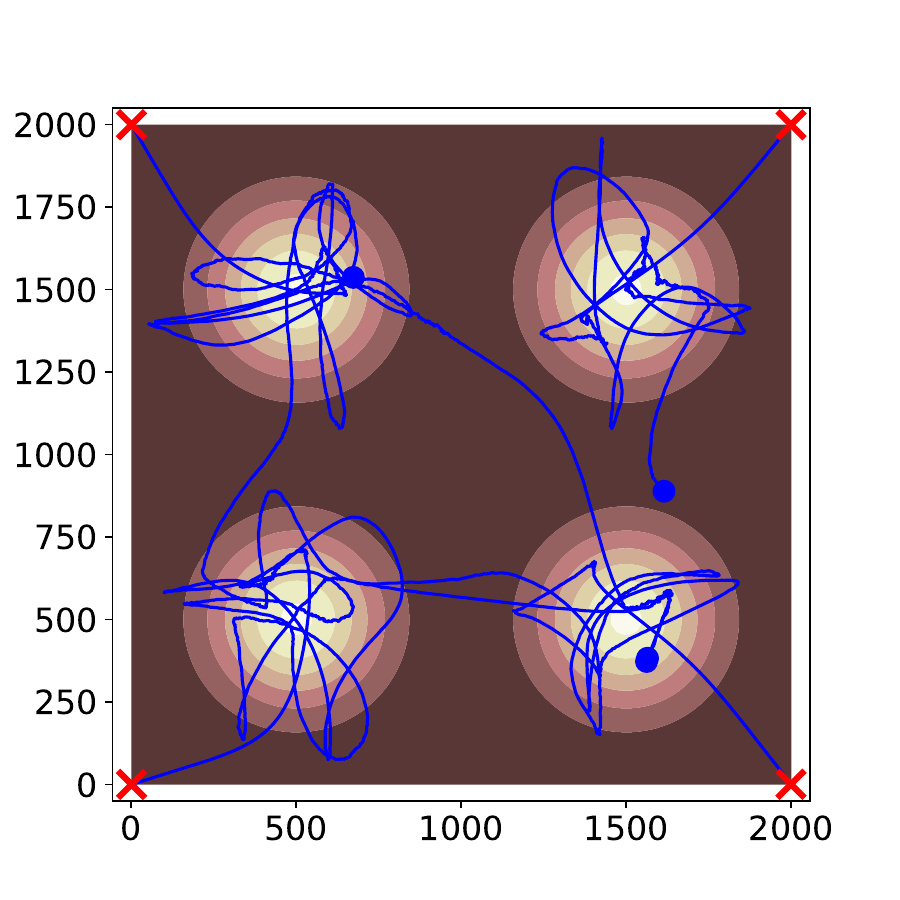}}
\subfloat[]{
\includegraphics[width=0.5\linewidth]{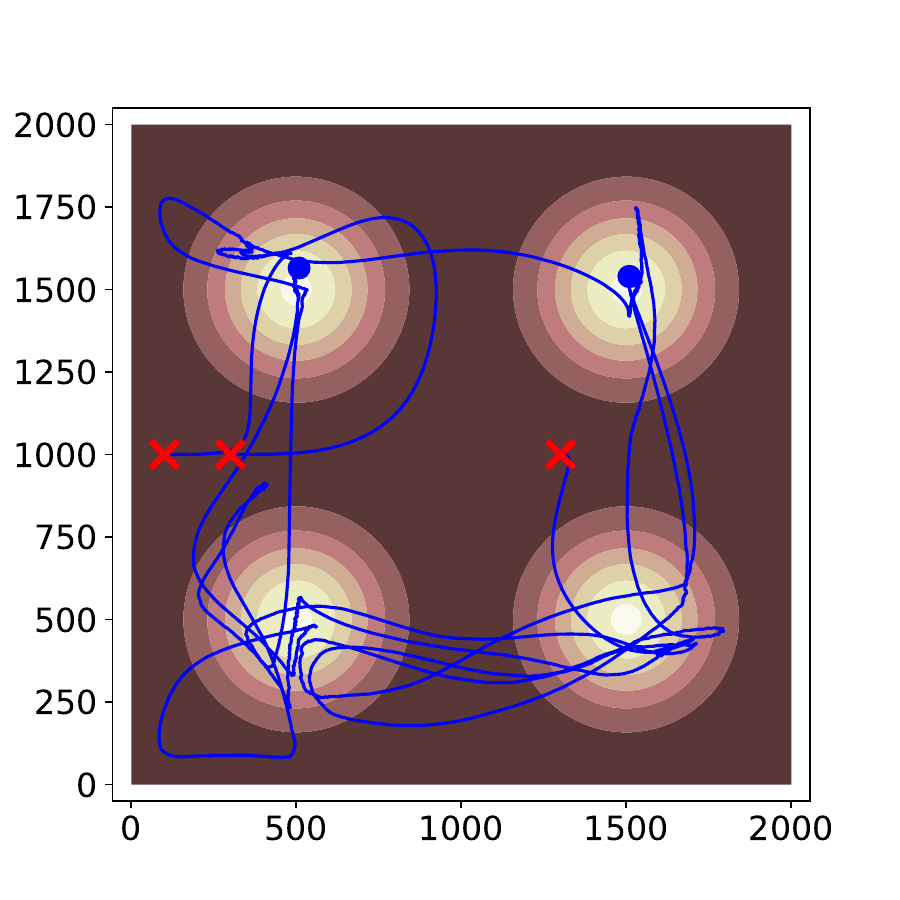}}
\caption{Agent trajectories under stochastic spectral control. The undesired early-stage directional bias is eliminated due to perturbation, and all agents remain bounded within the domain.}
\label{fig:corrected}
\end{figure}

\section{Conclusion and Future Work}

Classical SMC can exhibit directional degeneracy due to symmetry-induced invariant manifolds, leading to axis- or diagonal-constrained motion and transient stalling. This paper introduced stochastic perturbations with a contraction term to mitigate such behaviors, guaranteeing that agents almost surely escape zero-gradient manifolds, with mean-square boundedness of their trajectories preserved.

Future work includes extensions to heterogeneous agents, higher-dimensional domains, adaptive perturbation strategies, and experimental validation on multi-robot platforms.


\bibliographystyle{ieeetr} 
\bibliography{ref}

\end{document}